\documentclass[12pt]{article}

\usepackage{booktabs}
\usepackage{graphicx}%
\usepackage{multirow}%
\usepackage{amsmath,amssymb,amsfonts}%
\usepackage{amsthm}%
\usepackage{mathrsfs}%
\usepackage[title]{appendix}%
\usepackage{xcolor}%
\usepackage{textcomp}%
\usepackage{manyfoot}%
\usepackage{booktabs}%
\usepackage{algorithm}%
\usepackage{algorithmicx}%
\usepackage{algpseudocode}%
\usepackage{listings}%
\usepackage{url}%
\usepackage[margin=1in]{geometry}
\usepackage[width=0.8\textwidth]{caption}

\newtheorem{lemma}{Lemma}
\newtheorem{theorem}{Theorem}
\newtheorem{comment}{Comment}

\begin{document}

\title{Proof-of-Stake Dynamics: The Elusive Price Anchor and Endogenous Volatility Harvesting}

\author{
    Mikhail Perepelitsa \\
    University of Houston \\
    \small 3551 Cullen Blvd, PGH \\
    \small Houston, TX 77204-3008, USA
}

\date{}

\maketitle

\abstract{
In this paper, we develop an open-economy macroeconomic model of a Proof-of-Stake network to analyze nominal token-price dynamics and the systemic effects of speculative capital.

We first consider a network populated solely by active utility users, who finance network activity through a steady exogenous inflow of fiat currency. We prove the existence of a unique, globally asymptotically stable steady-state equilibrium with a well-defined nominal token price and derive a closed-form expression for the network's relaxation time. Calibrating the model using parameters representative of the current Ethereum network, we estimate a relaxation half-life of approximately 46 years. This extreme macroeconomic inertia implies that the token price may remain persistently displaced from its evolving steady-state benchmark, producing sustained price overshooting as the network adjusts to changing fundamentals.

We then introduce an Investor class to examine the effects of passive and active speculative capital. We show that passive institutional staking compresses the native staking yield and creates a structural imbalance that systematically raises the nominal token price while shifting consensus ownership away from active utility users. Active speculative capital has a qualitatively different effect. In response to capital shocks, the Consumer class's rigid preference for fiat-denominated consumption generates an endogenous constant-value strategy. This mechanism shifts staked-token ownership from the Investor class toward active utility users, with potentially favorable implications for consensus decentralization.
}


\section{Introduction}\label{sec:intro}

This paper develops an explanatory open-economy macroeconomic model to address two fundamental questions concerning Proof-of-Stake (PoS) networks. Contemporary macroeconomic discussions of Ethereum (ETH) are increasingly shaped by the network's rapid financialization. The expansion of liquid-staking protocols and the introduction of yield-bearing exchange-traded products have attracted substantial institutional capital, encouraging investors to treat ETH as a passive yield-bearing asset rather than primarily as a transactional utility token.

This structural transformation raises two central questions. First, is the absence of an observable macroeconomic equilibrium token price an intrinsic feature of the network, or is it a consequence of institutional speculation? Second, how does speculative capital affect the network's economic functioning, token distribution and issuance, and consensus security?

To answer the first question, we start with the model, called the Dual-Consumption Model of the economy of Proof-of-Stake networks, introduced in \cite{Perepelitsa2026_II} and strip away the speculative layer contained in it to isolate the network's intrinsic economic drivers. We consider a market populated initially only by Consumers: active utility users with well-defined preferences who periodically inject exogenous, real-world fiat to finance network transactions and staking. We analyze the behavior of the dynamical system at the center of this model for the existence and stability of the steady-states which represent the economic macroscopic equilibria.

To address the second question, we extend the model by adding a financial agent (the Investor) who passively keeps capital in the network and subjects the network to buy/sell shocks, allowing us to quantify the effects of these shocks on the system.

Our main findings are as follows.

\noindent\textit{The Price Anchor and Macroeconomic Inertia.}
We first prove that the Consumer-only economy possesses a unique, globally asymptotically stable steady-state equilibrium. At this equilibrium, the network satisfies two mass-balance conditions: newly issued tokens equal tokens burned, while external fiat inflows equal fiat liquidations. The resulting equilibrium nominal price therefore provides a well-defined long-run price anchor.

Our quantitative analysis, however, reveals substantial macroeconomic inertia. Under a calibration based on current Ethereum network parameters, the estimated relaxation half-life is approximately 50 years. Because economically relevant shocks and changes in network fundamentals occur on much shorter time scales, the system is unlikely to approach its prevailing steady state before that steady state shifts again. The network may therefore remain in persistent transition toward a moving equilibrium benchmark.

This mechanism is analogous to the overshooting dynamics of classical open-economy models \cite{Dornbusch1976,ObstfeldRogoff1996}. In both settings, the interaction between fast and slow variables generates prolonged departures from equilibrium. In our model, the nominal token price is the fast variable, responding immediately to changes in external conditions, whereas the physical token stock adjusts gradually through issuance, burning, and portfolio reallocation. Thus, persistent nominal-price disequilibrium can arise as an intrinsic consequence of the network's adjustment dynamics, even in the absence of speculative capital.

\noindent\textit{The Centralizing Effect of Passive Capital.}
In the second part of the paper, we examine the systemic effect of passive institutional staking. We prove that when a large Investor class holds and stakes tokens without contributing to transactional demand, its participation compresses the endogenous staking yield below the Consumer-only equilibrium level. The resulting imbalance causes the nominal token price to rise monotonically while the Consumer's share of the staked token stock declines monotonically.

Passive capital therefore moves the system away from the Consumer-only steady state and gradually reallocates consensus ownership from active utility users toward non-transactional Investors. This result identifies a centralizing macroeconomic channel that operates independently of any direct technological advantage enjoyed by institutional staking providers.

\noindent\textit{Endogenous Volatility Harvesting.}
We next use numerical simulations to examine the propagation of speculative capital shocks, modeled as periodic inflows into and outflows from network staking. We observe that the Consumer’s fixed propensity to consume in fiat generates an automated constant-value strategy: it sells tokens to Investors when speculative inflows raise the token price and repurchases them when subsequent outflows depress it. This response creates a Shannon’s-demon-like volatility-harvesting mechanism. Under a symmetric fiat cycle and frictionless market clearing, the Consumer’s accumulated fiat-denominated wealth remains unchanged, while physical tokens are systematically transferred from the Investor to the Consumer. Speculative volatility is therefore harvested in token units rather than in fiat.

Active speculative capital thus has an effect fundamentally different from passive staking. Although it amplifies short-run price volatility, it shifts staked-token ownership toward active utility users, with potentially favorable implications for consensus decentralization.

\subsection{Literature Overview}\label{subsec:lit_review}

The open-economy macroeconomic framework developed in this paper builds on Jermann's models of Ethereum under Proof-of-Stake and EIP-1559 \cite{jermann2025macro,jermann2025optimal}, which characterize its long-run equilibrium and connect token valuation, staking, issuance, and network activity. Related models derive token value from transactional demand and platform adoption \cite{cong2021tokenomics} and examine the interaction between utility users and speculative traders \cite{sockin2023model}. Our framework extends this literature by introducing distinct Consumer and Investor classes, explicit external fiat flows, and large deterministic capital shocks.

More specifically, we replace the homogeneous-agent and quasi-linear utility specifications with a nested utility structure that generates a nonzero wealth effect and a fixed marginal propensity to consume in fiat. This directly links nominal token-price dynamics to external capital flows. We also shift the focus from small stochastic disturbances to speculative purchase and liquidation cycles, allowing us to identify endogenous stabilizing and redistributive mechanisms. The underlying Consumer--Investor structure builds on the Dual-Consumption model introduced in \cite{Perepelitsa2026_II}; the present paper establishes the stability and relaxation properties of the Consumer-only equilibrium and distinguishes the effects of passive staking from those of active speculative capital.

Our analysis also relates to research on equilibrium staking and consensus ownership. Existing studies examine how heterogeneous investment horizons and staking rewards affect participation and redistribution \cite{john2021equilibrium}, as well as whether Proof-of-Stake reward mechanisms concentrate token holdings \cite{rosu2021evolution,fanti2019compounding}. Empirical work further documents the growing concentration of Ethereum's validating power among centralized and liquid-staking services \cite{grandjean2023ethereum}. Our model identifies a complementary mechanism through which external yield-seeking capital compresses staking yields and shifts consensus ownership away from transactional users.

Finally, Elowsson \cite{elowsson2021circulating} characterizes Ethereum's steady-state token supply by equating validator issuance with EIP-1559 token burn, but does not endogenize the nominal token price through behavioral agents and fiat flows. Our analysis of price adjustment instead draws an analogy with classical open-economy overshooting models \cite{Dornbusch1976,ObstfeldRogoff1996}: the nominal token price responds rapidly to capital shocks, whereas the staked-token stock adjusts slowly, producing persistent deviations from a moving steady-state benchmark.

\section{The Dual-Consumption Model}
\label{sec:utility_model}

In the Dual-Consumption Model,  the Consumer class represents the active utility users of the network. Unlike speculative actors, Consumers possess well-defined preferences for transacting on the network, staking tokens, and liquidating a portion of their holdings for outside-world consumption. Each period, they introduce a predictable, exogenous inflow of fiat ($I_c$) derived from their outside-world income to finance their ongoing network activities. The Consumer's economic behavior is rigidly governed by a hierarchical utility maximization framework. 

The first-level utility function $U_1(C,V)$ utilizes a standard logarithmic Cobb-Douglas specification to model the trade-off between immediate real-world fiat consumption $C$ and the aggregate capital volume $V$ retained within the digital asset ecosystem. This formulation yields fixed expenditure shares, parameterizing $\nu$ as the consumer's marginal propensity to consume out of total crypto-inclusive wealth $W_c$, and $\xi = 1-\nu$ as the crypto retention rate. Conditional on the top-level macro allocation $V$, the second-level utility function $U_2(S_c, L_c)$ governs the network portfolio optimization, splitting the consumer's crypto-fiat allocation between yield-bearing physical staked supply $S_c$ and active gas transaction utility $L_c$. We implement a quasi-linear utility structure at this second tier to capture an asymmetry in user behavior: while the demand for staking assets $S_c$ is highly elastic and scales linearly with wealth allocations, the demand for core transactional throughput $\gamma \log L_c$ is inelastic. In this formula, $\gamma>0$ measures the intensity of the Consumer’s demand for transactional token use.

Thus, the first-level utility is given by the formula

\begin{equation}
    U_1(C,V) = \nu\log C + (1-\nu)\log V,\quad \nu\in(0,1),
\end{equation}
with the budget constraint 
\[ 
C+V = W_c,
\]
and the second-level utility:
\begin{equation}
    U_2(S_c,L_c) = E[p]S_c(1+y) + \gamma\log L_c,
\end{equation}
with the budget
\[
pS_c + pL_c = V.
\]
Here, $E[p]$ is the next period token price. 
Optimizing the consumption we get:
\[
C= \nu W_c,\quad V = \xi W_c,\quad \xi=1-\nu.
\]
Optimizing the crypto balance, we obtain:
\begin{equation}
\label{def:L_C}
L_c=\frac{\gamma}{E[p](1+y)},
\end{equation}
\begin{equation}
\label{def:S_c}
S_c = V/p- L_c = \frac{\xi W_c}{p}-L_c.
\end{equation}
To simplify the analysis, we assume myopic price expectations, 
$E[p^{t+1}]=p^t.$

In contrast to the utility-driven Consumer, the Investor class represents purely speculative capital. These actors do not participate in the network's endogenous transactional utility, meaning they do not consume block space or incur transaction fees. Instead, they interact with the protocol strictly as a yield-bearing digital bond, passively staking their entire token balance to capture the native issuance. Within the Dual-Consumption framework, the Investor acts as the primary catalyst for macroeconomic volatility by injecting or extracting exogenous capital shocks ($\Lambda^t$) during speculative boom-and-bust cycles.

We summarize the state variables and parameters for the open-economy model that we use in this paper in Table \ref{tab:notation}.

\begin{table}[htbp]
    \centering
    \begin{tabular}{@{}l p{9.5cm}@{}}
        \toprule
        \textbf{Symbol} & \textbf{Description} \\
        \midrule
        $M$ & Total active physical supply of the native token (e.g., ETH). \\
        $S$ & Total physical tokens staked in the network. \\
        $p$ & Fiat-denominated unit price of the native token (e.g., USD price of ETH). \\
        $y$ & Native staking yield, defined as $y = \frac{c}{\sqrt{S}}$, where $c$ is the network issuance parameter. \\
        $W_c$ & Fiat-denominated budget of the Consumer. \\
        $I_c$ & Rate of Consumer fiat in-flow.\\
        $S_c$ & Physical tokens demanded for staking by the Consumer. \\
        $L_c$ & The physical quantity of tokens expended on network transactions during the period. \\
        $S_i$ & Physical tokens demanded for staking by the Investor. \\
        $\Lambda^t$ & Rate of Investor fiat in/out flow.\\
                \bottomrule
        \end{tabular}
    \caption{State variables and parameters for the Dual-Consumption Model.}
    \label{tab:notation}%
\end{table}


\subsubsection{Network supply evolution $M^{t+1}$}\label{subsubsec:supply2}

We denote by $M^t$ the active token stock represented in the model, comprising staked tokens and tokens allocated to transactional expenditure. Passive unstaked token balances that are neither staked nor allocated to current transactional expenditure are outside the modeled active stock.

The total active supply of ETH at period $t$ is split between tokens staked and tokens designated for burn:
\[ 
M^t = S^t + L^t_c.
\]
The number of tokens changes due to the network inflation and burn:
\[
M^{t+1} = S^t(1+y^t).
\]
The new amount of tokens staked is determined from equation
\[
M^{t+1} = S^{t+1} + L^{t+1}_c,
\]
where $L^{t+1}_c$ is the new demand for gas by the consumers. And finally,
\begin{equation}
\label{eq:S_t+1 II}
S^{t+1} = S^t(1+y^t) - L^{t+1}_c.
\end{equation}

\subsubsection{The Consumer and the Investor staking evolution}
\label{subsubsec:price2}

The total network staking is composed of the Investor and the Consumer staking amounts:
\[
S^t {}={} S^t_i + S^t_c,\quad t=0,1,2...
\]

We consider an open-loop model in which the Investor stakes all the wealth that exogenously changes with the in-flow rate of capital $\Lambda_t.$

The wealth accumulated by the Investor and the Consumer over one period:
\begin{itemize}
    \item Next Period Investor Staking Amount 
    \begin{equation}
    \label{Investor New Stake}
    S^{t+1}_i = S_i^t(1+y^t){}+{}\frac{\Lambda^{t+1}}{p^{t+1}}.
    \end{equation}
    \item Next Period Consumer Total Wealth (including exogenous fiat income $I_c^{t+1}$): 
    \begin{eqnarray}
    \label{Consumer New Wealth 1}
    W_c^{t+1} &=& p^{t+1} S_c^t (1 + y^t) + I^{t+1}_c\\
     \label{Consumer New Wealth 2}
             & = &p^{t+1} \left( \frac{\xi W_c^t}{p^t} - L_c^t \right) (1 + y^t) + I^{t+1}_c.
    \end{eqnarray}
\end{itemize}

\subsubsection{Clearing Price Equation}

To derive the next period price $p^{t+1}$, we begin with the market-clearing identity. The total token supply available for allocation at date $t+1$
must equal the sum of the Investor's new physical demand and the Consumer's new physical demand:
\begin{equation}
\label{eq:supply_demand_balance}
S^t(1+y^t) = S_i^t(1+y^t) + \frac{\Lambda^{t+1}}{p^{t+1}} + \frac{\xi W_c^{t+1}}{p^{t+1}}.
\end{equation}

Since $S^t(1+y^t) = S_c^t(1+y^t)+S_i^t(1+y^t),$ and using 
equation \eqref{Consumer New Wealth 1}, we get the Balance equation
\begin{equation}
\label{Sell Buy Balance}
(1-\xi)S^t_c(1+y^t) = \frac{\Lambda^{t+1}}{p^{t+1}} + \frac{\xi I_c^{t+1}}{p^{t+1}}, 
\end{equation}
where the left-hand side represents the number of tokens sold by the Consumer class and the right-hand side represents the amounts purchased by the Investor and the Consumer, respectively.

Using \eqref{def:S_c} on the left-hand side of \eqref{Sell Buy Balance}, and solving for $p^{t+1},$ we get,
\begin{equation}
\label{eq:price_clearing 0}
p^{t+1} = p^t \left[ \frac{\Lambda^{t+1} + \xi I_c^{t+1}}{\nu \left( \xi W_c^t - p^t L_c^t \right) (1+y^t)} \right].
\end{equation}
This formula can also be written in terms of the Consumer staking amount $S_c^t:$
\begin{equation}
\label{eq:price_clearing I}
p^{t+1} =  \frac{\Lambda^{t+1} + \xi I_c^{t+1}}{\nu S_c^t (1+y^t)}.
\end{equation}

Finally, we substitute the liquid gas demand constraint $L_c^t = \frac{\gamma}{p^t(1+y^t)}$ into the denominator to reach the fully expanded, actionable price equation:
\begin{equation}
\label{eq:price_clearing II}
p^{t+1} = p^t \left[ \frac{\Lambda^{t+1} + \xi I_c^{t+1}}{\nu \xi W_c^t (1+y^t) - \nu \gamma} \right].
\end{equation}

Substituting \eqref{eq:price_clearing 0} in \eqref{Consumer New Wealth 2} we observe the balance of the fiat inflow (Right) and outflow (Left):
\begin{equation}
    \label{eq:Fiat Balance}
    \nu W^{t+1}_c = \Lambda^{t+1} + I_c^{t+1}.
\end{equation}
Here and below, we restrict $\Lambda^{t}> -\xi I_c^t.$

For future use, we record 
\begin{itemize}
    \item Next Period Consumer Staking Amount
    \begin{align}    
        S^{t+1}_c &= \xi S_c^t(1+y^t) + \frac{\xi I_c^{t+1}}{p^{t+1}} - L^{t+1}_c\\
        \label{New Consumer Staking}
        &= \xi S_c^t(1+y^t) + \frac{1}{p^{t+1}} \left( \xi I_c^{t+1} - \frac{\gamma}{1+y^{t+1}} \right).
    \end{align}
\end{itemize}
\begin{comment}
    Since the Consumer fiat inflow rate $I_c^t$ is vastly larger than the gas burn fees $\frac{\gamma}{1+y^{t+1}}$ under the Ethereum calibration, equation \eqref{New Consumer Staking} shows that there is a lower bound on the rate of decrease of the Consumer staking level:
    \[
    S^{t+1}_c> \xi S_c^t
    \]
    no matter how hard the Investor is buying off tokens, $(\Lambda^t>0)$, preventing a fast takeover of the network by the Investor.
\end{comment}

\subsubsection{The Next-Period Staking Level}

Once the market-clearing price $p^{t+1}$ is established, it determines the required fiat-equivalent gas burn, which dictates the physical supply evolution. The law of motion for the physical staked supply states that the next period's staked tokens equal the newly inflated total supply minus the tokens burned for transactional utility:
\begin{equation}
\label{eq:next_S_evolution}
S^{t+1} = S^t(1+y^t) - L_c^{t+1}.
\end{equation}

We substitute the functional definitions for both the next-period gas demand $L_c^{t+1}$ and the next-period yield $y^{t+1}$:
\begin{equation}
L_c^{t+1} = \frac{\gamma}{p^{t+1}(1+y^{t+1})}, \quad \text{where} \quad y^{t+1} = \frac{c}{\sqrt{S^{t+1}}}.
\end{equation}

Plugging these into \eqref{eq:next_S_evolution}, yields $S^{t+1}$:
\begin{equation}
\label{Total New Staking}
S^{t+1} = S^t(1+\frac{c}{\sqrt{S^t}}) - \frac{\gamma}{p^{t+1} \left( 1 + \frac{c}{\sqrt{S^{t+1}}} \right)}.
\end{equation}


\subsection{Dynamical System Formalism}
For the study of the dynamics of the Dual-Consumption Model described in the previous sections, it will be convenient to represent the state of the network at period $t$ by a vector of the staking level of the Consumer $(S^t_c),$  total staked ETH in the network $(S^{t})$ and price $p^t:$
\[
(S_c^t,S^t,p^t).
\]
The values of all other variables are functions of this state. 
The change of the state vector over one period is governed by the system consisting of equations \eqref{eq:price_clearing I}, \eqref{New Consumer Staking} and \eqref{Total New Staking}.
This system implicitly defines the next period values from the current values and exogenous inputs:
\begin{equation}
    \label{def:DS}
   (S^{t+1}_c, S^{t+1},p^{t+1}) = \mathcal{\bf F}(S^t_c, S^t, p^t,\Lambda^{t+1}, I_c^{t+1}),
   \end{equation}
   for some vector-valued function $\mathcal{\bf F}.$
The model is completed by the choice of the initial values for 
\[
(S_i^0, S^0,p^0).
\]


\subsection{Steady-State Equilibrium for the Consumer-only Economy}
\label{subsec:steady2}
We consider first the Consumer-only Economy ($S^t_i=0,\Lambda^t=0$). In this case $S^t = S^t_c,$ and 
 $W_c^t = \frac{I^t_c}{\nu}.$ The dynamics is described by the equations
\begin{align}
    \label{def:DS Pure Consumer 1}
    p^{t+1} & =  \frac{\xi I_c^{t+1}}{\nu S^t (1+y^t)}, \\
    \label{def:DS Pure Consumer 2}
    S^{t+1} & =  S^t(1+y^t) - \frac{\gamma}{p^{t+1} \left( 1 + y^{t+1} \right)}.
\end{align}
The  Consumer dynamics admits a unique steady state if we assume that the Consumer fiat inflow is constant: $I_c^t=I_c.$ In such a case, the steady-state values for the price, the consumer wealth, the staking amount, and ETH yield can be obtained from the above equations, see \cite{Perepelitsa2026_I}:
\begin{equation}
\label{def:eq values}
S^* = \left( \frac{c \xi I_c}{\gamma \nu} \right)^2,\quad p^* = \frac{\gamma^2 \nu}{c^2 (\xi I_c + \gamma \nu)}, \quad y^*  = \frac{\gamma \nu}{\xi I_c},\quad W_c^* = \frac{I_c}{\nu}.
\end{equation}

The steady-state is the intersection point of  the equilibrium for fiat in-flow and out-flow rates:
\begin{equation}
    \label{def: consumer wealth}
    I_c = \nu W_c^*,
\end{equation}
the equilibrium between newly minted tokens and the tokens burned:
 \[
 S^*y^* = \frac{\gamma}{p^*(1+y^*)},
 \]
and the price  at which tokens  demanded and tokens liquidated are equal: 
\[
\frac{\xi I_c}{p^*} = \nu S^*(1+y^*).
\]

We will call the set of all steady-states parametrized by the fiat inflow rate $I_c$ the steady-state manifold.

\begin{comment}[Network Expansion]
The formula for the equilibrium price of ETH can be used to obtain a crude quantitative estimate of the effect of the expansion of the network. If there is $N$ participants in the economy, then both $\gamma \sim N$ and $I_c \sim N$ and from the formula for $p^*$ in \eqref{def:eq values} we conclude that the equilibrium price will grow linearly in $N$ during such expansion.
\end{comment}

\begin{comment}
\label{com:1}
    Notice that the equilibrium price $p^*$ in \eqref{def:eq values} is inversely proportional to fiat inflow rate $I_c,$ which we can informally identify as the Consumer demand for tokens. Increased demand does lead to an increase in the token price, but only in the short term. 
    The price jump triggers the issuance of tokens above the burn level, which over time brings the system back to the equilibrium manifold, with the lower equilibrium price $p^*.$ 
\end{comment}

\noindent\textbf{Numerical Example}
We illustrate convergence to equilibrium after an increase in demand by the following numerical example. We consider the system at a steady-state initially with values reflecting realistic network conditions, as summarized in Table \ref{tab:baseline_parameters}.

\begin{table}[htbp]
    \centering
    \begin{tabular}{p{0.55\linewidth} l c}
        \toprule
        \textbf{Description} & \textbf{Parameter} & \textbf{Value} \\
        \midrule
        Monthly periodic yield (3\% annualized target) & $y^*$ & $0.0025$ \\
        Initial physical staked supply (tokens) & $S^0$ & $40 \times 10^6$ \\
        Baseline exogenous monthly fiat inflow (USD) & $I_c^0$ & $350 \times 10^6$ \\
        Fiat consumption propensity & $\nu$ & $0.01$ \\
        Crypto retention rate & $\xi$ & $0.99$ \\
        Gas utility parameter & $\gamma$ & $86.625 \times 10^6$ \\
        Network issuance parameter & $c$ & $15.81$ \\
        \bottomrule
    \end{tabular}
     \caption{Baseline calibration parameters for the steady-state equilibrium.}
     \label{tab:baseline_parameters}
\end{table}

At $t=1$, we introduce a macroeconomic shock where the fiat inflow demand permanently increases by $50\%$, jumping to $I_c^1 = 525 \times 10^6$ USD.

Following the shock, the system is bound to an asset-market clearing locus derived from equations \eqref{def:DS Pure Consumer 1} and \eqref{def:DS Pure Consumer 2}:
\begin{equation}
\label{def:supply curve}
p = \frac{1}{S} \left( \frac{\xi I^1_c}{\nu} - \frac{\gamma}{1+y} \right).
\end{equation}
The system then moves along this curve until it reaches a new steady state (see Figure \ref{fig:Attractor}).

The dynamics in this example goes through the following causal sequence.

\textit{The Price Jump:} When the demand, $I_c$, jumps, the market-clearing price $p$ must spike instantaneously because the physical staked supply $S$ cannot jump.

\textit{The Burn Collapse:} With a significantly higher nominal token price, the physical tokens required to satisfy the Consumer's constant fiat gas demand decrease sharply from the previous steady-state level.
    
\textit{The Supply Expansion:} Because the physical burn rate collapses while the protocol's native yield continues to issue tokens, the network suddenly mints more tokens than it burns, leading to expansion of the physical staked supply $S.$

\textit{The Monotone Slide:} The system then enters a slow, multi-decade monotone sliding mode
    along the short-run equilibrium schedule \eqref{def:supply curve}. As more tokens are minted and the physical supply expands, the nominal price incrementally decreases. This decreasing price causes the burned tokens to increase,  see \eqref{def:L_C}, until the burn rate finally catches up to perfectly balance the minting rate at the new steady-state equilibrium.

The stability and slow adjustment rate in this example are generic properties of the network. We address these properties in the next section.


\subsubsection{Stability of the Steady State and the Rate of  Relaxation}

In the case of the Consumer-only Economy, the dynamical system \eqref{def:DS} reduces to 1-dimensional system for the staking amount $S^{t}.$
The non-linear equation for the next period $S^{t+1}$ is given by
\begin{equation}
    \label{def:DS 1-d}
    S^{t+1} = S^t(1+y^t)\left( 1 - \frac{y^*}{1+ y^{t+1}}\right),\quad
    y^t = \frac{c}{\sqrt{S^t}},
\end{equation}
which implicitly determines $S^{t+1} = F(S^t).$

Our main result, which we prove in  Section \ref{sec:appendix}, is that the steady state is unique and is globally and asymptotically stable.
\begin{theorem}[Global Asymptotic Stability of the Consumer-only Economy]
\label{th:1}
Consider the deterministic discrete-time dynamical system governing the physical staked supply in a Consumer-only economy, defined implicitly by the transition map $S^{t+1} = F(S^t)$ from equation
\begin{equation}
    S^{t+1} = S^t(1+y(S^t))\left( 1 - \frac{y^*}{1+ y(S^{t+1})}\right)
\end{equation}
where $y(S) = \frac{c}{\sqrt{S}}$ and the steady-state baseline yield is $y^* = \frac{\nu \gamma}{\xi I_c}$. Assume standard macroeconomic conditions such that $0 < y^* < 1$. Then, the unique steady state $S^*$ is globally asymptotically stable. Specifically:
\begin{enumerate}
    \item[i.] For any arbitrary initial physical supply $S^0 > 0$, the state sequence $\{S^t\}_{t=0}^{\infty}$ converges monotonically to the steady state $S^*$.
    \item [ii.] Sequence $S^t$ converges to $S^*$ at the relaxation rate $\lambda,$
    \begin{equation}
    \label{def:Relaxation Rate}
       \lim_{t\to\infty}\frac{S^{t+1}-S^*}{S^t-S^*} = \lambda = \frac{1 + \frac{1}{2}y^*}{1 + y^* + \frac{1}{2}(y^*)^2}  \approx 1 - \frac{1}{2}y^*.
    \end{equation}
\end{enumerate}
\end{theorem}

\vspace{0.5cm}
\noindent\textbf{Numerical Example} To illustrate the macroeconomic inertia of the network in reality, we evaluate the convergence rate using the baseline parameters from Table \ref{tab:baseline_parameters}. According to Theorem \ref{th:1} the relaxation rate of the physical supply $S^t$ and the price $p^t$ given by 
\begin{equation*}
    \lambda \approx 1 - \frac{1}{2}y^*  = 0.99875.
\end{equation*}

This value of $\lambda$ is exceptionally close to $1.$ The macroeconomic half-life ($t_{1/2}$), which represents the time required for a network shock to dissipate by exactly 50\%, 
is
\begin{equation*}
    t_{1/2} = \frac{\ln(0.5)}{\ln(\lambda)}  \approx 554.17 \text{ months}.
\end{equation*}
This translates to an adjustment half-life of approximately 46 years. 

This calculation shows that while the Consumer-only economy is globally and asymptotically stable, it takes decades for the physical consensus layer to move the system to the steady-state manifold.

The steady-state is entirely determined by the exogenous fiat inflow, $I_c$. In reality, this inflow is highly dynamic, fluctuating in response to real-world business cycles, shifts in market sentiment, and technological innovation. These traditional economic drivers operate on relatively short timescales—business cycles typically unfold over a few years, while broader macroeconomic trends can alter fiat liquidity in a matter of months. Because these real-world fluctuations happen faster than the network's 46-year relaxation time, the physical layer is structurally too slow to ever catch up. The network is repeatedly hit by new shocks long before it can resolve the previous ones. Consequently, the system exists in a perpetual state of transition, constantly chasing an ever-changing target on the attractor. This dynamic closely mirrors the classic Asset Price Overshooting model introduced by Dornbusch \cite{Dornbusch1976}, and further formalized by Obstfeld and Rogoff \cite{ObstfeldRogoff1996}. Just as in those established frameworks, our network features a fast variable (the nominal price $p^t$) that instantaneously absorbs any change in $I_c$ (Equation \eqref{eq:price_clearing I}), alongside a slow variable (the physical supply $S^t$) that drags the system toward the new steady state at a crawling pace.

\begin{figure}[htbp]
    \centering
    \includegraphics[width=0.8\textwidth]{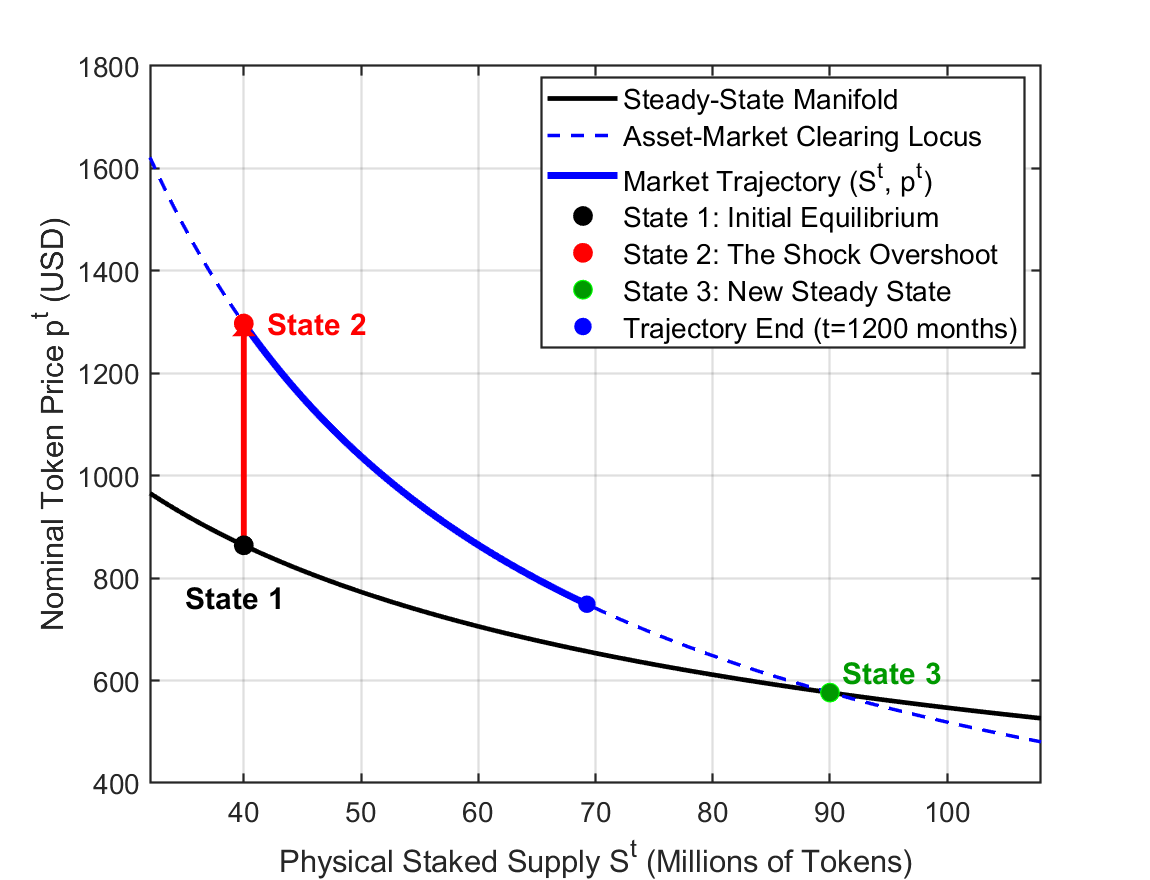}
    \caption{Simulated dynamics of the Consumer-only Market. At time $t=0$, the market is at equilibrium State 1. At $t=1$ the fiat inflow rate $I_c$ increases by 50\%. The fast variable $p^t$ immediately adjusts to the new level (State 2). The network proceeds slowly to the steady-state equilibrium State 3. The simulation stops after $t=1,200$ months, or 100 years.}
    \label{fig:Attractor}
\end{figure}

\section{Consumer/Investor Dynamics}
In this section, we will extend the Consumer-only Economy with an additional financial sector represented by the Investor who keeps capital in ETH staking. The financial goals of the Investor are outside of the scope of our model. Here we will focus on two simple scenarios. (A) Passive Accumulation: the Investor has a certain amount of staked ETH that remains in the network for a long period of time. (B) Buy/Sell Shocks: The Investor who has a substantial amount of staked ETH occasionally buys or sells large amounts of tokens.

In both scenarios, the question is how the network reacts to the addition of the Investor class. Specifically, we will be interested in the (de)centralization: the share of tokens owned by the Consumer and the Investor, and the nominal price dynamics.  

We start with scenario (A).
\begin{lemma}[Effects of Passive Investor Accumulation]
\label{lem:passive_accumulation}
Assume that the Consumer fiat inflow is constant,
\[
I_c^t=I_c>0,
\]
and that the Investor makes no additional purchases or liquidations after the
initial date,
\[
\Lambda^t=0,
\qquad t\geq 1.
\]
Let the equilibrium yield
$
y^*=\frac{\nu\gamma}{\xi I_c},\quad 
0<y^*<1.
$
Then the following statements hold.
\begin{enumerate}
    \item[\rm (i)]
    For any initial Consumer wealth \(W_c^0>0\), the Consumer's
    fiat-denominated wealth satisfies
    \[
    W_c^t=\frac{I_c}{\nu},
    \qquad t\geq 1.
    \]

    \item[\rm (ii)]
    Suppose that at \(t=1\) the state is interior and satisfies
    \[
    0<S_c^1<S^1,
    \qquad
    0<y^1<y^*,
    \qquad
    \frac{S_c^1}{S^1}<\frac{y^1}{y^*}.
    \]
Then, for every \(t\geq 1\),
    \[
    S^{t+1}>S^t,
    \qquad
    0<y^{t+1}<y^t<y^*,
    \]
    \[
    p^{t+1}>p^t,
    \qquad
    S_c^{t+1}<S_c^t,
    \qquad
    S_i^{t+1}>S_i^t.
    \]
    Consequently, the Consumer's share of the total staked-token stock
    decreases monotonically:
    \[
    \frac{S_c^{t+1}}{S^{t+1}}
    <
    \frac{S_c^t}{S^t},
    \qquad t\geq 1.
    \]

    The exact one-period price-inflation rate is
    \[
    \frac{p^{t+1}-p^t}{p^t}
    =
    \frac{y^*-y^t}{1-(y^*-y^t)}
    =
    y^*-y^t
    +
    \mathcal{O}\!\left((y^*-y^t)^2\right).
    \]
\end{enumerate}
\end{lemma}

\begin{proof}
Under passive accumulation, \(\Lambda^{t+1}=0\). The fiat-balance
identity gives
\[
\nu W_c^{t+1}=I_c,
\]
and hence
\[
W_c^t=\frac{I_c}{\nu},
\qquad t\geq1.
\]
This proves part \({\rm (i)}\).

For \(t\geq1\), substituting \(W_c^t=I_c/\nu\) and
\[
p^tL_c^t=\frac{\gamma}{1+y^t}
\]
into the market-clearing price equation gives
\[
\frac{p^{t+1}}{p^t}
=
\frac{\xi I_c}
{\xi I_c(1+y^t)-\nu\gamma}
=
\frac{1}{1+y^t-y^*}
=
\frac{1}{1-(y^*-y^t)}.
\]
Therefore, whenever \(y^t<y^*\),
\[
p^{t+1}>p^t,
\]
and
\[
\frac{p^{t+1}-p^t}{p^t}
=
\frac{y^*-y^t}{1-(y^*-y^t)}.
\]

We next determine when total stake increases. Using
\[
p^{t+1}
=
\frac{\xi I_c}
{\nu S_c^t(1+y^t)}
\]
and \(y^*=\nu\gamma/(\xi I_c)\), the next-period burn is
\[
L_c^{t+1}
=
\frac{\gamma}{p^{t+1}(1+y^{t+1})}
=
\frac{y^*S_c^t(1+y^t)}{1+y^{t+1}}.
\]
The total-stake equation therefore becomes
\[
S^{t+1}
=
S^t(1+y^t)
-
\frac{y^*S_c^t(1+y^t)}{1+y^{t+1}}.
\]
We claim that  $
y^tS^t>y^*S_c^t$ implies $S^{t+1}>S^t.$
Since \(y(S)=c/\sqrt{S}\), this implies
\[
y^{t+1}<y^t.
\]
Suppose, to the contrary, that
\(S^{t+1}\leq S^t\). Since \(y(S)=c/\sqrt{S}\), this implies
\(y^{t+1}\geq y^t\). Therefore,
\[
\frac{y^*S_c^t(1+y^t)}{1+y^{t+1}}
\leq y^*S_c^t.
\]
Using the total-stake equation, we obtain
\[
S^{t+1}
\geq
S^t(1+y^t)-y^*S_c^t
=
S^t+\bigl(y^tS^t-y^*S_c^t\bigr)
>
S^t,
\]
which is a contradiction.

For \(t\geq1\), the Consumer's staked-token position can be written as
\[
S_c^t
=
\frac{1}{p^t}
\left(
\frac{\xi I_c}{\nu}
-
\frac{\gamma}{1+y^t}
\right)
=
\frac{\xi I_c}{\nu p^t}
\left(
1-\frac{y^*}{1+y^t}
\right).
\]
If \(y^{t+1}<y^t\), then
\[
1-\frac{y^*}{1+y^{t+1}}
<
1-\frac{y^*}{1+y^t}.
\]
Together with \(p^{t+1}>p^t\), this gives
\[
S_c^{t+1}<S_c^t.
\]

It remains to verify that the sufficient condition
\[
y^tS^t>y^*S_c^t
\]
is preserved. Since \(S^{t+1}>S^t\) and \(S_c^{t+1}<S_c^t\),
\[
y^{t+1}S^{t+1}
=
c\sqrt{S^{t+1}}
>
c\sqrt{S^t}
=
y^tS^t
>
y^*S_c^t
>
y^*S_c^{t+1}.
\]
Thus,
\[
y^{t+1}S^{t+1}>y^*S_c^{t+1},
\]
so the argument applies inductively for every \(t\geq1\).

Finally, because \(\Lambda^{t+1}=0\),
\[
S_i^{t+1}=S_i^t(1+y^t)>S_i^t.
\]
Since \(S_c^t\) decreases while \(S^t\) increases,
\[
\frac{S_c^{t+1}}{S^{t+1}}
<
\frac{S_c^t}{S^t}.
\]
This proves part \({\rm (ii)}\).
\end{proof}

The lemma shows that when the network yield is below the Consumer-only equilibrium level and the Consumer’s share
of total staking lies below the state-dependent threshold ($y^t/y^*$), the system enters a self-reinforcing regime: total stake rises, the staking yield declines, the nominal token price increases monotonically, and the Consumer’s share of staked tokens decreases monotonically. Passive Investor accumulation therefore moves the system progressively away from the Consumer-only steady-state manifold.

\subsection{Buy--Sell Shocks}

We consider now the Investor/Consumer Economy in the scenario (B).
For the presentation we opt for a numerical simulation that captures the main features of the market dynamics. In our model, the Investor is passively staked in ETH at the beginning of the simulation and then injects capital, $\Lambda^t>0,$ during a short period (Buy Shock), and later liquidates the same amount of capital, $\Lambda^t<0,$ during a period of the same length.

The network is parameterized according to the conditions in Table \ref{tab:initial_conditions}.

\begin{table}[htbp]
    \centering
    \begin{tabular}{llc}
        \toprule
        \textbf{Parameter} & \textbf{Symbol} & \textbf{Value} \\
        \midrule
        Fiat Inflow (Monthly) & $I_c$ & 350 Million USD \\
        Fiat Consumption Propensity & $\nu$ & 1\% \\
        Baseline Staking Yield (Monthly) & $y^*$ & 0.25\% \\
        Initial Staked Supply & $S^0$ & 40 Million ETH \\
        Initial Consumer Staked Supply & $S_c^0$ & 10 Million ETH \\
        Initial Investor Staked Supply & $S_i^0$ & 30 Million ETH \\
        Initial Token Price & $p^0$ & $\approx$ 3,456 USD \\
        Initial Consumer Wealth & $W_c^0$ & 35 Billion USD \\
        Gas Utility Parameter & $\gamma$ & $86.625 \times 10^6$ \\
        Network Issuance Parameter & $c$ & $15.81$ \\
        \bottomrule
    \end{tabular}
    \caption{Market Model Parameters for the Dual-Consumption simulation.}
    \label{tab:initial_conditions}
\end{table}

The simulation modeled a symmetric shock: the Investor injected $\Lambda^t=300$ million USD monthly from months $t=20,...,30,$ and subsequently extracted the exact same fiat amount, $\Lambda^t = -300$ million USD from months 40 to 50.

The results of the simulation are summarized in Figure \ref{fig:Buy/Sell}. The main characteristic features of the dynamics are the following.

\textit{Price Overshooting and Persistent Displacement:}
The nominal token price ($p^t$) rises to nearly 7,000 USD during the buy shock and falls to approximately 300 USD during the sell shock. Following the cycle, the simulated price recovers to approximately $2,100$ over the displayed post-shock horizon, remaining below its initial value during that interval; see the top-left panel of Figure \ref{fig:Buy/Sell}.

\textit{Physical Token Redistribution:}
The Consumer endogenously accumulates staked tokens, increasing $S_c$ from 10 million to approximately 17 million ETH. The Investor's staked-token holdings decline by the corresponding amount; see the bottom-left panel of Figure \ref{fig:Buy/Sell}.

\textit{Fiat-Denominated Wealth:}
The cumulative net Investor fiat flow over the complete cycle is zero: the Investor withdraws the same total amount of fiat that it previously injected. Nevertheless, the Investor ends the cycle with fewer tokens. The Consumer's total accumulated wealth, defined as the total fiat extracted minus total fiat injected, added to the Consumer's current token wealth,
\[
 \left( \sum_{k=1}^{t} (\nu W_c^k - I_c^k) \right) + W_c^t,
\]
returns to its pre-shock value because $\sum_{k=1}^{t} (\nu W_c^k - I_c^k) = \sum_{k=1}^t \Lambda^k =0,$
for $t>50.$

\subsubsection{The asymmetries in the Buy/Sell cycle.}

The market mechanics of the Buy/Sell cycle produce visible asymmetries in the dynamics of all market variables.

\textit{The Buy Shock:} When the Investor injects fiat into the token market, the price spikes instantly. Because they are purchasing tokens at the elevated token price ($p^t \approx 7,000$), their 300 million USD buys a relatively small volume of physical tokens.

The Consumer allocates the fixed share ($\nu W_c^t$) of its wealth to fiat consumption. As the price increase raises the fiat-denominated value of the Consumer's token holdings, its desired consumption rises. When ($\nu W_c^t>I_c^t$), the Consumer finances the difference by liquidating tokens, thereby supplying the tokens purchased by the Investor.

\textit{The Sell Shock:} During the sell shock, the Investor liquidates tokens, while the Consumer provides the corresponding buying demand required by market clearing.

The price crashes ($p^t \approx 300$). The fixed fiat withdrawal requires the Investor to liquidate a substantially larger volume of tokens than they originally purchased. 

At a lower token price, a given fiat surplus allows the Consumer to acquire a larger physical quantity of tokens. The Consumer algorithmically absorbs this large token liquidation at the depressed price, resulting in the persistent post-shock increase in $S_c$. 

In this cycle, the Consumer effectively executes a constant-value rebalancing strategy: algorithmically selling "High" and buying "Low."  The asymmetry of the Buy--Sell cycle follows from the market-clearing identity
\begin{equation}
\label{def:constant_value}
I_c^t-\nu W_c^t=-\Lambda^t,\qquad t>0.
\end{equation}
In the absence of Investor flows, $\Lambda^t=0,$ this identity yields $W_c^t/I_c^t=1/\nu.$ We therefore refer to $1/\nu$ as the Consumer's target Constant-Value Ratio. Investor flows temporarily displace the Consumer from this target and induce the compensating token sales and purchases described below.

The left-hand side of \eqref{def:constant_value} represents the Consumer's net fiat contribution to the token market: its exogenous fiat inflow ($I_c^t$) less its fiat consumption ($\nu W_c^t$). Thus, the Consumer's net order flow exactly offsets the exogenous Investor flow.

During the buy shock, ($\Lambda^t>0$), and hence ($\nu W_c^t>I_c^t$). The Consumer therefore liquidates tokens to finance consumption in excess of its current fiat inflow, providing the token supply purchased by the Investor as the price rises. During the sell shock, ($\Lambda^t<0$), and hence ($\nu W_c^t<I_c^t$). The Consumer then retains a net fiat surplus and purchases tokens liquidated by the Investor.

Because a fixed fiat amount purchases fewer tokens at the elevated buy-shock price and more tokens at the depressed sell-shock price, a symmetric fiat Buy--Sell cycle generates an asymmetric transfer of physical tokens from the Investor to the Consumer. The Consumer's allocation rule therefore produces an endogenous contrarian rebalancing mechanism: it sells tokens during the price increase and purchases tokens during the subsequent decline.

The simulation also exhibits persistent finite-horizon displacement in the nominal token price. Following the Buy–Sell cycle, the price remains below its initial value over the displayed post-shock period. This displacement is associated with a substantial redistribution of staked tokens between the two classes: the Consumer holds a larger post-shock token balance, while its fiat-denominated transactional demand remains unchanged.

The token-denominated burn rate rises sharply during the sell shock, see Figure \ref{fig:Buy/Sell}, bottom right display. In the model, the lower token price increases the quantity of ETH associated with a given level of fiat-denominated transactional demand; the burn spike therefore need not represent a comparable increase in real network activity.

\subsubsection{Financial and consensus-ownership implications}

The symmetric Buy–Sell cycle leaves the Consumer’s accumulated fiat-denominated wealth unchanged while increasing its physical token holdings. The cycle is therefore wealth-neutral and token-accretive for the Consumer under the stated accounting measure.

From the perspective of consensus ownership, the cycle shifts staked tokens from the non-transactional Investor to the active utility-driven Consumer. This redistribution may have favorable implications for the alignment of validating power with network use. Establishing a corresponding improvement in decentralization or consensus security, however, would require an explicit measure of validator concentration or attack cost.

\begin{figure}[htbp]
    \centering
    \includegraphics[width=\textwidth]{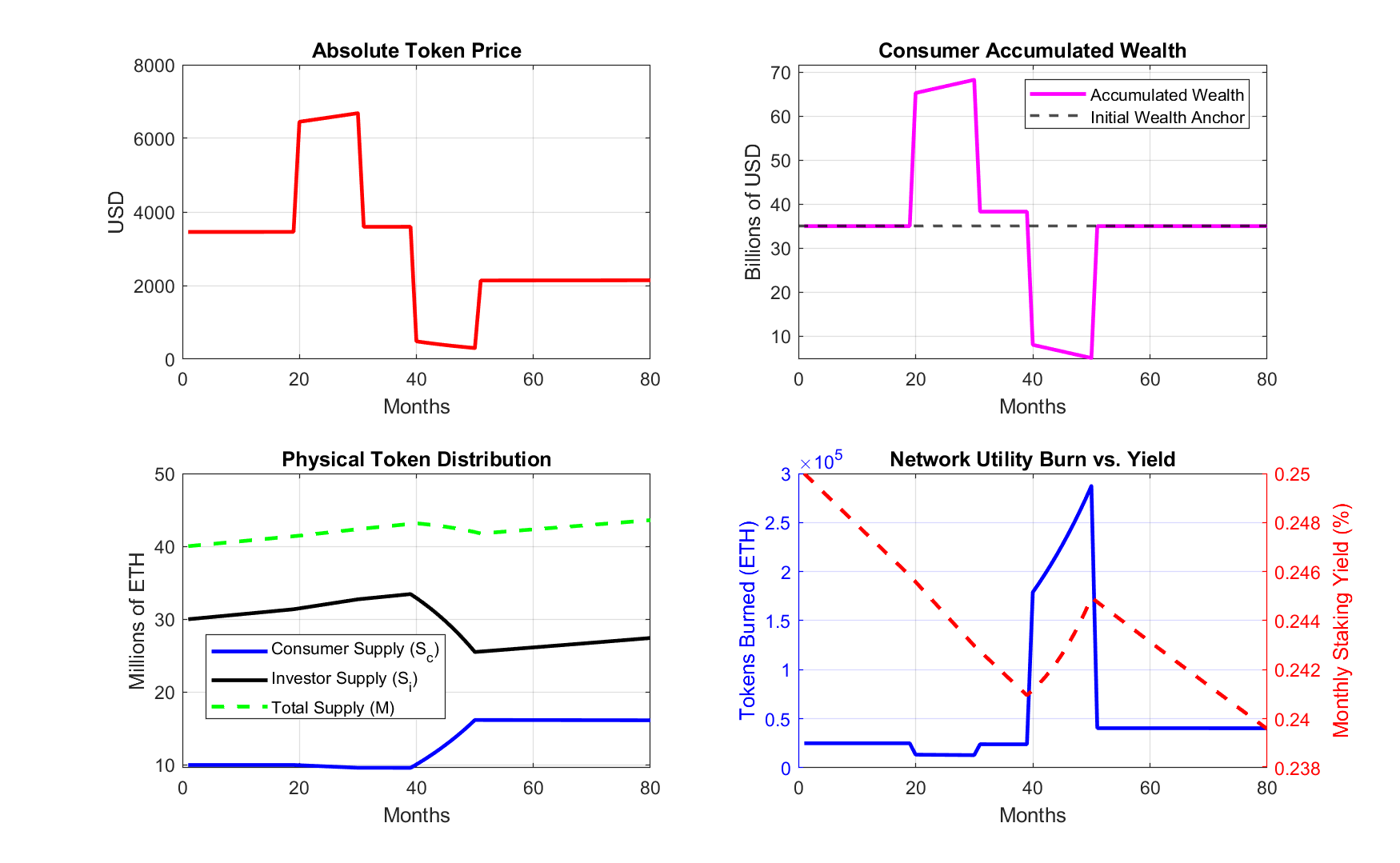}
    \caption{Simulated dynamics of a Buy-Sell cycle. The Investor is passively staking tokens in the first 20 periods. At periods $t=20,...,30$ the Investor invests 300 million USD per period into staking. At periods $t=40,...,50$ the Investor liquidates the same amount, 300 million USD per period.
    }
    \label{fig:Buy/Sell}
\end{figure}


\section{Conclusions and Limitations}
\label{sec:conclusion}

This paper develops a dynamical open-economy model of a Proof-of-Stake network with utility-driven Consumers and non-transactional Investors. In the Consumer-only economy, we establish the existence of a unique globally asymptotically stable steady state that provides a well-defined nominal price anchor. The calibrated local relaxation half-life is approximately 46 years, indicating that the physical token stock adjusts much more slowly than the nominal price. Consequently, the observed market price may remain persistently displaced from an equilibrium benchmark that itself changes with network fundamentals.

The introduction of Investor capital produces two distinct effects. Under the sufficient conditions identified in Lemma~\ref{lem:passive_accumulation}, passive restaking suppresses the native yield, raises the nominal token price, and gradually shifts the staked-token share away from active Consumers. By contrast, symmetric Buy--Sell shocks generate an endogenous contrarian response. The Consumer sells tokens during speculative inflows and repurchases them during subsequent outflows. Under frictionless market clearing, this Shannon's-demon-like mechanism leaves the Consumer's accumulated fiat-denominated wealth unchanged while increasing its physical token holdings. Speculative volatility is therefore harvested in token units rather than in fiat.

The model is intentionally stylized. It represents Consumers and Investors as aggregate classes, assumes myopic price expectations, and abstracts from transaction costs, slippage, taxes, staking delays, leverage, heterogeneous beliefs, and strategic Investor behavior. The Buy--Sell experiment uses deterministic symmetric fiat shocks; asymmetric or stochastic capital flows need not produce the same redistribution. The calibration is illustrative rather than a formal empirical estimation, and the model's active token stock excludes passive unstaked balances. Finally, a transfer of stake toward the Consumer class does not by itself establish greater validator decentralization or consensus security. Evaluating those effects would require explicit measures of ownership concentration, delegation structure, and attack cost. Another natural extension is to characterize the volatility-harvesting mechanism analytically for general Buy--Sell cycles and to determine its robustness to asymmetric capital flows, transaction costs, and stochastic price dynamics.

\section{Appendix}
\label{sec:appendix}

\subsection{Proof of Theorem 1}

\begin{proof}
To prove part i., let $S^*$ be the steady-state of \eqref{def:DS 1-d} and $1-y^*>0.$ Given $S^0>0$ we will show that $S^t\to S^*,$ as $t\to\infty.$
We restructure the implicit recurrence relation into the form 
\[
H(S^{t+1}) = R(S^t),
\]
where the auxiliary functions are defined as:
\begin{equation}
    R(x) = x(1+y(x)), \quad \text{and} \quad H(x) = \frac{x(1+y(x))}{1+y(x)-y^*},
\end{equation}
where the yield curve $y(x) = cx^{-1/2}.$ 
Direct computation shows that both $R(\cdot)$ and $H(\cdot)$ are strictly increasing functions.

Since $H(x)$ is strictly increasing, its continuous inverse $H^{-1}(x)$ exists, allowing us to explicitly define the transition map as $F(x) \equiv H^{-1}(R(x))$. As the composition of two strictly increasing functions, $F(x)$ is strictly increasing.

To determine the direction of the transition map relative to the line $y=x$, we evaluate the sign of $R(x) - H(x)$:
\begin{equation}
    R(x) - H(x) = x(1+y(x)) \left( \frac{y(x) - y^*}{1 + y(x) - y^*} \right)
\end{equation}
Because the native yield $y(x)$ is strictly monotonically decreasing,
\begin{itemize}
    \item for $x < S^*$, $y(x) > y^*$. The difference $(y(x) - y^*)$ is strictly positive, ensuring $R(x) > H(x)$, which implies $F(x) > x;$
    \item for $x > S^*$, $y(x) < y^*$. The difference $(y(x) - y^*)$ is strictly negative, ensuring $R(x) < H(x)$, which implies $F(x) < x$.
\end{itemize}
Since $F(x)$ is strictly increasing, sequences generated by $S^{t+1} = F(S^t)$ are monotonic. If $S^0 < S^*$, the sequence is strictly increasing and bounded above by $S^*$. If $S^0 > S^*$, it is strictly decreasing and bounded below by $S^*$. By the Monotone Convergence Theorem, for any $S^0 > 0$, $\lim_{t \to \infty} S^t = S^*$.

The linearized relaxation rate $\lambda$ in part ii. is defined as $F'(S^*).$ The value of $F'(S^*)$ is obtained by differentiating equation \eqref{def:DS 1-d} and evaluating it at the steady state $S^*:$
\[
F'(S^*) = \frac{1 + \frac{1}{2}y^*}{1 + y^* + \frac{1}{2}(y^*)^2} = 1 - \frac{1}{2}y^* + \mathcal{O}((y^*)^2).
\]
\end{proof}

\bibliographystyle{plain} 
\bibliography{refs}

\end{document}